\documentclass[12pt]{article}

\usepackage[T1]{fontenc}
\usepackage{kpfonts}
\usepackage{pdflscape}
\usepackage{setspace}
\usepackage[margin=1.25in]{geometry}

\usepackage[dvipsnames]{xcolor}

\usepackage{pdfpages}
\usepackage{float}
\usepackage{multirow}
\usepackage{caption}
\usepackage{subcaption}
\usepackage{epigraph}

\usepackage{mathtools}
\usepackage{amsthm}
\usepackage{aliascnt} 
\usepackage{accents}
\usepackage{dutchcal}
\usepackage{bbm}

\usepackage{enumitem}
\usepackage{sgame}

\usepackage{tikz}
\usepackage{tikz-cd}
\usetikzlibrary{calc,shapes,arrows}

\usepackage{tcolorbox}
\usepackage{listings}

\usepackage[normalem]{ulem}

\usepackage[round]{natbib}
\newtheorem{theorem}{Theorem}[section]

\newaliascnt{proposition}{theorem}
\newtheorem{proposition}[proposition]{Proposition}
\aliascntresetthe{proposition}

\newaliascnt{lemma}{theorem}
\newtheorem{lemma}[lemma]{Lemma}
\aliascntresetthe{lemma}

\newaliascnt{corollary}{theorem}

\aliascntresetthe{corollary}

\newaliascnt{claim}{theorem}
\newtheorem{claim}[claim]{Claim}
\aliascntresetthe{claim}

\theoremstyle{definition}

\newaliascnt{definition}{theorem}
\newtheorem{definition}[definition]{Definition}
\aliascntresetthe{definition}

\newaliascnt{example}{theorem}

\aliascntresetthe{example}

\newaliascnt{assumption}{theorem}

\aliascntresetthe{assumption}

\newaliascnt{condition}{theorem}

\aliascntresetthe{condition}

\newaliascnt{question}{theorem}

\aliascntresetthe{question}

\newaliascnt{remark}{theorem}

\aliascntresetthe{remark}

\newaliascnt{remarks}{theorem}

\aliascntresetthe{remarks}

\newaliascnt{aside}{theorem}

\aliascntresetthe{aside}

\newaliascnt{note}{theorem}

\aliascntresetthe{note}

\usepackage{thmtools}
\usepackage{thm-restate}

\usepackage{hyperref}

\hypersetup{
    colorlinks=true,
    linkcolor=Aquamarine,
    filecolor=Thistle,
    urlcolor=Thistle,
    citecolor=Thistle,
}

\usepackage[nameinlink]{cleveref}

\crefname{theorem}{theorem}{theorems}
\Crefname{theorem}{Theorem}{Theorems}

\crefname{proposition}{proposition}{propositions}
\Crefname{proposition}{Proposition}{Propositions}

\crefname{lemma}{lemma}{lemmas}
\Crefname{lemma}{Lemma}{Lemmas}

\crefname{corollary}{corollary}{corollaries}
\Crefname{corollary}{Corollary}{Corollaries}

\crefname{claim}{claim}{claims}
\Crefname{claim}{Claim}{Claims}

\crefname{definition}{definition}{definitions}
\Crefname{definition}{Definition}{Definitions}

\crefname{example}{example}{examples}
\Crefname{example}{Example}{Examples}

\crefname{assumption}{assumption}{assumptions}
\Crefname{assumption}{Assumption}{Assumptions}

\makeatletter
\let\cref@old@isrefconsecutive\cref@isrefconsecutive

\def\cref@isrefconsecutive#1#2{%
  \begingroup
    \def\cref@assumptiontype{assumption}%
    \cref@gettype{#1}{\cref@typea}%
    \ifx\cref@typea\cref@assumptiontype
      \endgroup
      \@cref@refconsecutivefalse
    \else
      \endgroup
      \cref@old@isrefconsecutive{#1}{#2}%
    \fi
}
\makeatother

\crefname{condition}{condition}{conditions}
\Crefname{condition}{Condition}{Conditions}

\crefname{question}{question}{questions}
\Crefname{question}{Question}{Questions}

\crefname{remark}{remark}{remarks}
\Crefname{remark}{Remark}{Remarks}

\crefname{remarks}{remarks}{remarks}
\Crefname{remarks}{Remarks}{Remarks}

\crefname{aside}{aside}{asides}
\Crefname{aside}{Aside}{Asides}

\crefname{note}{note}{notes}
\Crefname{note}{Note}{Notes}

\crefname{appendix}{appendix}{appendices}
\Crefname{appendix}{Appendix}{Appendices}

\newenvironment{introepigraph}{%
  \vspace{0.5em}
  \begin{center}
  \begin{minipage}{0.78\textwidth}
  \hrule height 0.4pt
  \vspace{0.7em}
  \footnotesize
  \setlength{\parindent}{0pt}
  \setlength{\parskip}{0.55em}
}{%
  \vspace{0.35em}
  \hrule height 0.4pt
  \end{minipage}
  \end{center}
  \vspace{1.1em}
}

\newcommand{\epigraphsource}[1]{%
  \par\vspace{0.2em}
  {\normalfont\footnotesize\hfill---\,#1}
}

\newcommand{\secref}[1]{\hyperref[#1]{\S\ref*{#1}}}

\definecolor{backcolour}{rgb}{0.63, 0.79, 0.95}

\lstdefinestyle{mystyle}{
  backgroundcolor=\color{backcolour},
  basicstyle=\ttfamily\footnotesize,
  breakatwhitespace=false,
  breaklines=true,
  captionpos=b,
  keepspaces=true,
  numbers=left,
  numbersep=5pt,
  showspaces=false,
  showstringspaces=false,
  showtabs=false,
  tabsize=2
}

\begin{document}
\author{Cooper Howes \and Can Urgun \and Mark Whitmeyer\thanks{CH: Federal Reserve Board of Governors, \href{cooper.a.howes@frb.gov}{ cooper.a.howes@frb.gov}, CU: \href{curgun@unc.edu}{curgun@unc.edu}, \& MW: Arizona State University, \href{mailto:mark.whitmeyer@gmail.com}{mark.whitmeyer@gmail.com}. We thank a conference audience at the 2025 SEA Annual Meetings for their feedback. The views expressed in this paper are solely those of the authors. They do not necessarily reflect the views of the Federal Reserve Board or the Federal Reserve System.}}

\title{Delusions of Grandeur and Their Benefits (and Hazards)}
\maketitle

\begin{abstract}
We study a population-wide tournament in which agents, who care both about their absolute and relative wealths, experiment by searching over correlated objects. We explore the role of the agents' beliefs about the environment; namely, the stochastic processes corresponding to their experimentation. We find that although optimism leads to higher output, it also produces greater inequality. We connect these observations with empirical evidence suggesting a positive relationship between inequality and entrepreneurship.
\end{abstract}

\begin{introepigraph}
The lottery of the law, therefore, is very far from being a perfectly fair
lottery; and that, as well as many other liberal and honourable professions,
is, in point of pecuniary gain, evidently under-recompenced.

Those professions keep their level, however, with other occupations, and,
notwithstanding these discouragements, all the most generous and liberal spirits
are eager to crowd into them.  Two different causes contribute to recommend
them.  First, the desire of the reputation which attends upon superior excellence
in any of them; and, secondly, the natural confidence which every man has more
or less, not only in his own abilities, but in his own good fortune.

\ldots

The over-weening conceit which the greater part of men have of their own
abilities, is an ancient evil remarked by the philosophers and moralists of all
ages.  Their absurd presumption in their own good fortune, has been less taken
notice of.  It is, however, if possible, still more universal.  There is no man
living who, when in tolerable health and spirits, has not some share of it.  The
chance of gain is by every man more or less over-valued, and the chance of loss
is by most men under-valued.

\epigraphsource{\citet[Book I, Ch.~X]{Smith1776}}
\end{introepigraph}

\section{Introduction}

Every entrepreneur understands the odds in the third person and rejects them in the first--every new venture is an argument that its founder is special. The product will find demand, the experiment will work, the market will arrive, or the next iteration will be the one that turns the corner. The social consequences of these delusions of grandeur are not obvious. The same belief that leads an entrepreneur to persist too long may also be the belief that keeps her searching long enough to discover something valuable.

We study the ramifications of entrepreneurial overconfidence in a model of competition by innovators. Our answer has two parts. Overconfidence raises output because it makes entrepreneurs harder to discourage. They continue after bad news, and so they find better opportunities. But overconfidence also raises inequality. The additional persistence generated by optimistic beliefs does not improve all outcomes equally, raising the upper tail by more than the rest of the distribution. Thus, \textbf{delusions of grandeur are productive, but their benefits are unevenly distributed.}

In our model, there is a continuum of \textit{ex-ante} identical entrepreneurs. Each entrepreneur experiments over time and may stop whenever she chooses. If she stops, she implements the best option she has found so far. Continuing is costly, and payoffs depend both on absolute success and on relative success: an entrepreneur values
the quality of her own best discovery, but she also values her rank in the population. The distribution of other entrepreneurs' outcomes is, therefore, an equilibrium object and determines how valuable a given discovery is in relative terms.

Crucially, our model also incorporates a simple asymmetry in beliefs. In actuality, no entrepreneur is special: all entrepreneurs face the same true search technology. Subjectively, however, each entrepreneur believes that her own project is special: formally, each entrepreneur believes that her own search process has positive drift, even though the true drift is zero. Moreover, she knows that other entrepreneurs also think they are special and that their decisions reflect this confidence. But she believes that their confidence is misplaced, while her own is warranted. Everyone thinks she is special, everyone knows that everyone else thinks she is special, and everyone thinks that everyone else is wrong.\footnote{This premise
is consistent with evidence on optimism, overconfidence, and self-assessed
ability among entrepreneurs and entrants; see, among others,
\citet{CamererLovallo1999}, \citet{DeMezaSouthey1996},
\citet{ArabsheibaniDeMezaMaloneyPearson2000},
\citet{KoellingerMinnitiSchade2007}, \citet{PuriRobinson2007},
\citet{DawsonDeMezaHenleyArabsheibani2014}, and
\citet{KoudstaalSloofVanPraag2016}.}

The decision problem each entrepreneur solves (at equilibrium) has a transparent form. An entrepreneur does not stop simply because the current state of her project is bad--after all, she can always fall back on the best option she has already discovered. What matters is how far the current project has fallen below that personal best, and so an entrepreneur keeps going until this disappointment becomes too large. Equilibrium behavior is summarized by a tolerance for bad news: how far below her best-so-far outcome an entrepreneur is willing to let the project fall before quitting.

\textbf{Overconfidence raises this tolerance.} A more optimistic entrepreneur thinks the future is more promising, so she is willing to absorb a larger setback before giving up. In equilibrium, this happens throughout the distribution, with more optimistic entrepreneurs searching longer, reaching better personal records, and stopping with better discoveries. Our first main result is, thus, holding fixed the true technology, a more overconfident economy generates a better distribution of realized discoveries (in a first-order stochastic dominance sense). Accordingly, \textbf{delusions of grandeur are good}: they raise output and fan the flames of technological growth.

Our second main result is distributional. In particular, the improvement from overconfidence is not a parallel shift. Extra persistence is most valuable for entrepreneurs who end up high in the realized distribution. These entrepreneurs have already
survived many earlier setbacks, and the additional willingness to continue helps them push the frontier even farther. Consequently, overconfidence strictly increases the variance of realized outcomes. \textbf{Delusions of grandeur are bad}: they lead to greater inequality.

It is important to keep in mind that optimists do not discover more because the true process is better--it is the same in every economy. Optimists discover more because they quit later. Overconfidence acts like a private subsidy to persistence, making discouraging evidence feel less decisive. This raises the average quality of discoveries, but it also increases the time spent experimenting and makes realized success more unequal.

Our empirical section takes two implications of the model to the data. The first concerns outcomes. If entrepreneurial experimentation raises discoveries while stretching the upper tail, then economies with more entrepreneurship should tend to exhibit more inequality. Consistent with this prediction, we document a positive cross-country relationship between entrepreneurship and inequality, using both the Gini coefficient and the 90/10 percentile ratio. The relationship remains positive after controlling for GDP per capita and population.

The second implication concerns the premise of the model. Our mechanism is not that entrepreneurs simply have a taste for risk but that they hold optimistic beliefs about their own prospects. We find little evidence that country-level measures of risk tolerance, uncertainty avoidance, equity-market participation, lottery sales, fear of failure, or business failure rates explain entrepreneurship. These variables also do not
remove the positive relationship between entrepreneurship and inequality. In contrast, survey measures of entrepreneurial self-belief are strongly related to entrepreneurship. In particular, the share of respondents who say that they have the ability to start a business is the strongest predictor in our cross-country exercises. The evidence is only correlational, but it supports the interpretation
that entrepreneurship is tied more closely to beliefs about one's own prospects than to (intrinsic) risk appetite.\footnote{This use of survey beliefs connects the paper to the broader survey-based literature on economic perceptions, social position, fairness, and policy views, including
\citet{KuziemkoNortonSaezStantcheva2015},
\citet{AlesinaStantchevaTeso2018}, \citet{Stantcheva2021TaxPolicy},
\citet{HvidbergKreinerStantcheva2023}, and \citet{Stantcheva2023Surveys}.}

We organize the rest of the paper as follows. \secref{sec:related}
discusses related work. \secref{sec:model} introduces the model.
\secref{sec:eq-characterization} characterizes equilibrium. \secref{sec:comparative-statics}
studies how overconfidence changes output, inequality, and experimentation time.
\secref{sec:empirics} presents the empirical evidence. Proofs of results lie in \Cref{omittedproofs,omittedproofsdeux}.

\subsection{Related Work}\label{sec:related}

There is by now a rich literature studying contests in which agents compete by choosing fair gambles,\footnote{See, for instance, \citet{FangNoe2016,FangNoe2022}, \citet{WagmanConitzer2012}, \citet{BoleslavskyCotton2015,BoleslavskyCotton2018}, \citet{Albrecht2017}, \citet{AuKawai2020}, and \citet{Spiegler2006}. In \citet{Whitmeyer2026DistributionalCompetition}, the gambles need not be fair.} or, equivalently,\footnote{The Skorokhod embedding theorem is the bridge equating the two.} when to stop Brownian motions.\footnote{This problem was first studied by \citet{SeelStrack2013}, with later extensions by \citet{FengHobson2015,FengHobson2016a,FengHobson2016b}, \citet{NutzZhang2022}, and \citet{Whitmeyer2023SubmissionCosts}.} A crucial distinction between these papers and ours is that upon stopping, our agents leave with the maximum value their processes hit, whereas in these papers, it is the stopped value of the processes they end with. The Brownian environment corresponds to a search environment with a continuum of correlated items; and in many settings, the ``perfect recall'' of our environment seems more reasonable than the ``no recall'' of these papers.
A related set of search-contest and research-tournament papers studies
environments in which contestants sequentially sample independent opportunities
rather than explore a correlated path.\footnote{See \citet{Taylor1995},
\citet{Whitmeyer2018}, \citet{ChenChenKnyazev2022},
\citet{HudjaRobersonRosokha2025}, and \citet{ErkurtOzdenoren2026}. These papers study independent-draw or independent-sampling search contests, with varying recall conventions.} Our paper is also part of a sizeable literature on incentives for risk taking in tournaments, R\&D races, promotion contests, and portfolio-management contests.\footnote{See, among others, \citet{DasguptaStiglitz1980}, \citet{BhattacharyaMookherjee1986}, \citet{KletteDeMeza1986}, \citet{Hvide2002}, \citet{HvideKristiansen2003}, \citet{GoelThakor2008}, \citet{Gilpatric2009}, \citet{FangNoeStrack2020}, \citet{BasakMakarov2015}, and \citet{Strack2016}.}

The closest technical antecedents of our work are models of search on a correlated path.\footnote{Refer to, in particular, \citet{UrgunYariv2025ContiguousSearch} and \citet{CetemenUrgunYariv2023CollectiveProgress}.} In these papers, an agent explores a Brownian or Gaussian environment and later implements the best observation encountered. We share the retrospective-best payoff, but embed the individual stopping problem in a population-wide game, making the agents' payoffs endogenous. The paper is also connected to trial-and-error learning and correlated search, where local experimentation reveals information about nearby alternatives.\footnote{See \citet{JovanovicRob1990}, \citet{Callander2011TrialError}, \citet{Callander2011GoodPolicies}, \citet{CallanderMatouschek2019Risk}, \citet{CallanderMatouschek2022Novelty}, and \citet{CallanderLambertMatouschek2025Rugged}.  On correlated attributes, selective learning, and broader correlated-learning environments, see \citet{Bardhi2024Attributes} and \citet{BardhiCallander2026Learning}.}


The paper also speaks to work on optimism, overconfidence, and entrepreneurship.  Existing models emphasize entry, financing, managerial selection, and entrepreneurial choice.  Our mechanism is different: optimism operates not through who enters, but through how long a fixed population continues experimenting.  This distinction is important for welfare.  Deluded beliefs can be productive because they sustain search, but hazardous because the gains accrue disproportionately to agents who end up high in the realized distribution.\footnote{On optimism and overconfidence in economic choice, entry, entrepreneurship, and management, see \citet{CamererLovallo1999}, \citet{DeMezaSouthey1996}, \citet{BernardoWelch2001}, \citet{PuriRobinson2007}, \citet{GoelThakor2008}, and \citet{LandierThesmar2009}.  Related models of motivated or optimal beliefs include \citet{BenabouTirole2002} and \citet{BrunnermeierParker2005}.}

Empirically, the paper is connected to evidence that entrepreneurship is associated with optimistic or overconfident beliefs, and that standard risk-preference measures do not fully account for entrepreneurial behavior.  This matches our mechanism: optimism matters because it changes the willingness to continue after bad realizations, not simply because agents have a taste for risk.\footnote{Evidence on optimism, overconfidence, and entrepreneurship includes \citet{ArabsheibaniDeMezaMaloneyPearson2000}, \citet{KoellingerMinnitiSchade2007}, \citet{PuriRobinson2007}, \citet{DawsonDeMezaHenleyArabsheibani2014}, \citet{Astebro2003}, \citet{AstebroHerzNandaWeber2014}, and \citet{KoudstaalSloofVanPraag2016}.}  The paper is also related to empirical work on entrepreneurial returns and experimentation.  Low average returns, high idiosyncratic risk, and skewed outcomes are consistent with a view of entrepreneurship as experimentation with large upside realizations.\footnote{See \citet{Hamilton2000}, \citet{MoskowitzVissingJorgensen2002}, \citet{KerrNandaRhodesKropf2014}, \citet{Manso2016}, and \citet{LevineRubinstein2017}.}

Finally, our empirical discussion relates to the literature on entrepreneurship and wealth inequality.  Much of that literature studies borrowing constraints, entrepreneurial saving, and business ownership as sources of wealth concentration.  We offer a complementary belief-based channel: a society with more entrepreneurial experimentation can generate both higher output and more unequal outcomes, even without ex ante heterogeneity in ability or wealth.\footnote{On entrepreneurship, wealth, and inequality, see \citet{EvansJovanovic1989}, \citet{Quadrini2000}, \citet{HurstLusardi2004}, and \citet{CagettiDeNardi2006}.}

\section{Model}\label{sec:model}

There is a continuum of experimenting agents of measure one. Experimentation takes the form of costly search over options: each agent runs an independent driftless Brownian process \(X_t=\sigma W_t\), where \(W\) is a standard Brownian motion and \(\sigma>0\). Let \(M_t\coloneqq\sup_{0\le s\le t}X_s\) denote the best option discovered by the agent by time \(t\). Agents are misspecified: for perceived drift \(\mu>0\), each agent evaluates stopping times under the perceived law \(\mathbb P^\mu\), under which \(X_t=\mu t+\sigma W_t^\mu\). The parameter \(\mu\) is the level of \textit{optimism}.

Experimentation is costly: each agent incurs a common flow cost \(c>0\). At any time, an agent may stop experimenting, in which case she receives the best option she has seen so far, \(M_t\). Agents care about both the absolute quality of the best option they find and their rank in the population. Agents evaluate rank through a conjectured rank schedule: a cumulative distribution function \(K\) on \([0,\infty)\). Given \(K\), define \(\phi_K(m)\coloneqq m+y\left(K(m)\right)\), where \(y\in C^2\left([0,1]\right)\), \(y'>0\), \(y''\le0\), and \(y'\) is bounded. An agent chooses a stopping time \(\tau\) to solve
\[
\sup_\tau\mathbb E^\mu\left[\phi_K(M_\tau)-c\tau\right].
\]
At equilibrium, the conjectured schedule \(K\) equals the realized distribution of stopped maxima under the true law.

For the verification arguments below, it is useful to describe the problem from an arbitrary current position. If the agent's current best option is \(m\), her current option can be described as \(m-z\), where \(z\) is the current \emph{drawdown}. Thus, as a random variable, the drawdown is \(D_t\coloneqq M_t-X_t\). Starting from \((m,z)\) means \((M_0,D_0)=(m,z)\), and we write \(\mathbb E^\mu_{m,z}\) for expectation under the perceived law from this starting point. We say a stopping time \(\tau\) is \emph{\(K\)-admissible} from \((m,z)\) if \(\mathbb E^\mu_{m,z}\left[\tau\right]<\infty\) and \(\mathbb E^\mu_{m,z}\left[\phi_K(M_\tau)^+\right]<\infty\). In our model, agents begin from \((0,0)\). In the fixed-\(K\) stopping problem, we take optimality over \(K\)-admissible stopping times.

Next, define \(\bar d_\mu\coloneqq\sigma^2\log\left(c/(c-\mu)\right)/(2\mu)\), and let
\[
\mathcal M\coloneqq\left\{\mu\in(0,c)\colon\frac{\mu}{c}\left(1+\frac{y'(0)}{\bar d_\mu}\right)<1\right\}.
\]
We make the standing assumption \(\mu\in\mathcal M\), which is a finite-stopping condition. It rules out cases in which the perceived drift and the marginal rank reward are so large that the agent wants to continue forever. 

A \emph{drawdown boundary} is a positive function \(d\colon[0,\infty)\to(0,\infty)\). Given \(d\), define \[\tau_d\coloneqq\inf\left\{t\ge0\colon M_t-X_t\ge d\left(M_t\right)\right\},\]
and write
\[
C_d\coloneqq \left\{\left(m,z\right)\colon m\ge 0,\ 0<z<d\left(m\right)\right\}.
\]

For a conjecture \(K\) with \(\phi_K\in C^1\), write
\[
V^\mu_K\left(m,z\right)\coloneqq \sup_\tau \mathbb E^\mu_{m,z}\left[\phi_K\left(M_\tau\right)-c\tau\right],
\]

where the supremum is over stopping times that are \(K\)-admissible from \(\left(m,z\right)\). A function \(V\) is \emph{regular} for \(d\) with terminal payoff \(\psi\) at perceived drift \(\mu\) if \(V\in C^{1,2}\left(C_d\right)\), extends continuously with the relevant one-sided derivatives to the boundaries, satisfies
\[
\frac{\sigma^2}{2}V_{zz}\left(m,z\right)-\mu V_z\left(m,z\right)-c=0,\qquad \forall \left(m,z\right)\in C_d,
\]
and obeys \emph{value matching} \(V\left(m,d\left(m\right)\right)=\psi\left(m\right)\), \emph{smooth fit} \(V_z\left(m,d\left(m\right)^-\right)=0\), and the \emph{reflected-boundary condition} \(V_m\left(m,0^+\right)+V_z\left(m,0^+\right)=0\).

Given \(K\), the boundary \(d\) is a \emph{regular optimal drawdown boundary} for the fixed-\(K\) problem at perceived drift \(\mu\) if \(\tau_d\) attains \(V^\mu_K\left(m,z\right)\) for every \(\left(m,z\right)\) with \(0\le z<d\left(m\right)\), among stopping times that are \(K\)-admissible from \(\left(m,z\right)\), and \(V^\mu_K\) is regular for \(d\) with terminal payoff \(\phi_K\) at perceived drift \(\mu\).
\begin{definition}
A \emph{regular drawdown equilibrium} at perceived drift \(\mu\in\mathcal M\) is a pair \(\left(K,d\right)\) with the following properties. The conjecture \(K\) is a cumulative distribution function on \(\left[0,\infty\right)\) such that \(\phi_K\in C^1\). The boundary \(d\) is \(C^1\), positive, bounded, bounded away from zero, and has a finite limit. Given \(K\), the boundary \(d\) is a regular optimal drawdown boundary for the fixed-\(K\) problem and \(K\) is consistent with the agents' equilibrium stopping rule under the true law: \(K\left(m\right)=P^0\left(M_{\tau_d}\le m\right)\).
\end{definition}

\section{Equilibrium Characterization}\label{sec:eq-characterization}

The point of this section is to turn the equilibrium fixed point into a one-dimensional object. The difficulty is that the two equilibrium objects discipline each other. A drawdown boundary \(d\) determines, under the true law \(\mathbb P^0\), how far agents climb before stopping and, hence, what the realized rank distribution is. At the same time, the conjectured rank schedule \(K\) determines, under the perceived law \(\mathbb P^\mu\), how valuable it is for an individual agent to continue searching. The characterization below separates these two roles before putting them back together.

We begin with the distributional side. Suppose agents stop when the current gap between the best option found and the current option reaches \(d(M_t)\). Conditional on having reached a record \(m\), a larger value of \(d(m)\) makes it less likely that the agent stops before reaching still higher records. In fact, over a small interval \([m,m+\Delta]\), the stopping probability is approximately \(\Delta/d(m)\). The survival probability of the stopped maximum is the accumulation of these local exit rates.

Toward stating our first result, given \(d\), define \(A_d(m)\coloneqq\int_0^m\frac{1}{d(r)}dr\), and denote \(T_m\coloneqq\inf\left\{t\ge0\colon M_t\ge m\right\}\).

\begin{lemma}\label{lem:survival} Let drawdown \(d\in C^1\left([0,\infty)\right)\) be positive and locally bounded away from zero. Then \(\mathbb P^0\left(T_m<\tau_d\right)=\exp\left(-A_d(m)\right)\). If \(A_d(m)\to\infty\) as \(m\to\infty\), then \(\tau_d<\infty\) \(\mathbb P^0\)-almost surely and \(\mathbb P^0\left(M_{\tau_d}\ge m\right)=\exp\left(-A_d(m)\right)\). Moreover, \(K_d(m)\coloneqq1-\exp\left(-A_d(m)\right)\) is the induced CDF of \(M_{\tau_d}\).
\end{lemma}

The proof is in \Cref{app:firstlem}. The lemma says that the inverse drawdown \(1/d(m)\) is the local rate at which agents exit before pushing their maximum past \(m\). Thus, a drawdown boundary is not only an individual stopping rule; under the true law, it is also a recipe for the cross-sectional rank distribution.

We next turn to the individual optimality side. Fix a conjectured rank schedule \(K\). If an agent stops with maximum \(m\), her effective prize is \(\phi_K(m)=m+y\left(K(m)\right)\). In \Cref{app:aux}, we show that, within the regular drawdown class, optimality for this fixed-\(K\) problem is characterized by the free-boundary equation
\[
d'(m)=1-\frac{\mu}{c}\frac{\phi_K'(m)}{1-\exp\left(-\frac{2\mu}{\sigma^2}d(m)\right)}.
\]
When \(K\) is differentiable, \(\phi_K'(m)=1+y'\left(K(m)\right)K'(m)\). This object decomposes in a straightforward manner: the term \(1\) is the direct marginal value of raising the best option found. The second term \(y'\left(K(m)\right)K'(m)\) is the marginal rank value: a slightly higher maximum improves the agent's position in the population, and this is valuable because \(y'>0\). This is the channel through which relative concerns affect the stopping boundary.

The limiting boundary has a simple interpretation. When the marginal rank term vanishes, the agent behaves like a single-agent experimenter whose marginal prize from a higher maximum is just one. A finite limiting drawdown must then converge to \(\bar d_\mu\), where
\[
1-\exp\left(-\frac{2\mu}{\sigma^2}\bar d_\mu\right)=\frac{\mu}{c}.
\]
This value will be the terminal condition for the equilibrium characterization.

We now impose the consistency necessitated by equilibrium. That is to say, in equilibrium, the conjectured rank schedule must equal the true distribution generated by the optimal boundary: \(K=K_d\). From \Cref{lem:survival}, this means \(K(m)=1-\exp\left(-A_d(m)\right)\). Rather than solve directly for both \(K\) and \(d\), we use \(t=A_d(m)\) as the coordinate. In this coordinate, the agent with maximum \(m(t)\) has rank \(K(m(t))=1-e^{-t}\). Define the quantile-indexed drawdown \(\delta(t)\coloneqq d(m(t))\). Since \(A_d(m(t))=t\), we differentiate to obtain \(m'(t)/d(m(t))=1\), or simply \(m'(t)=\delta(t)\). Consequently, the same function that records the drawdown tolerated at each rank also specifies the slope of the equilibrium quantile function.

The rank term also simplifies after this change of variables. Since \(K(m(t))=1-e^{-t}\), we can again differentiate to get \(K'(m(t))m'(t)=e^{-t}\). Using \(m'(t)=\delta(t)\), we obtain \(K'(m(t))=e^{-t}/\delta(t)\). Therefore, the marginal prize at the maximum \(m(t)\) is
\[
\phi_K'(m(t))=1+\frac{b(t)}{\delta(t)},\qquad \text{with} \qquad b(t)\coloneqq e^{-t}y'\left(1-e^{-t}\right).
\]
Substituting this expression into the fixed-\(K\) boundary equation at \(m=m(t)\) and using \(\delta'(t)=d'(m(t))m'(t)=d'(m(t))\delta(t)\), any regular equilibrium must satisfy
\[
\delta'(t)=\delta(t)-\frac{\mu}{c}\frac{\delta(t)+b(t)}{1-\exp\left(-\frac{2\mu}{\sigma^2}\delta(t)\right)}.
\]
This is the key reduction. The equilibrium fixed point in the two objects \(K\) and \(d\) becomes a terminal-value problem for the single function \(\delta\). Our theorem shows that this reduction is complete: the equation has a unique bounded positive solution, and that solution reconstructs the unique regular drawdown equilibrium.

\begin{theorem}\label{thm:equilibrium-characterization} Fix \(\mu\in\mathcal M\). There exists a unique regular drawdown equilibrium at perceived drift \(\mu\). It is characterized by the unique bounded positive \(C^1\) solution \(\delta_\mu\) of
\[
\delta_\mu'(t)=\delta_\mu(t)-\frac{\mu}{c}\frac{\delta_\mu(t)+b(t)}{1-\exp\left(-\frac{2\mu}{\sigma^2}\delta_\mu(t)\right)}
\]
satisfying \(\lim_{t\to\infty}\delta_\mu(t)=\bar d_\mu\), where \(b(t)\coloneqq e^{-t}y'\left(1-e^{-t}\right)\). Moreover, if \(\alpha_\mu\coloneqq1+y'(0)/\bar d_\mu\) and \(d_\mu^+\coloneqq\sigma^2\log\left(c/(c-\mu\alpha_\mu)\right)/(2\mu)\), then \(\bar d_\mu\le\delta_\mu(t)\le d_\mu^+\) for all \(t\ge0\). The equilibrium quantile map \(m_\mu(t)\coloneqq K_\mu^{-1}\left(1-e^{-t}\right)\) satisfies \(m_\mu'(t)=\delta_\mu(t)\).
\end{theorem}
Equivalently, once \(\delta_\mu\) is known, we can recover the equilibrium by setting \(m_\mu(t)=\int_0^t\delta_\mu(s)ds\), \(K_\mu(m_\mu(t))=1-e^{-t}\), and \(d_\mu(m_\mu(t))=\delta_\mu(t)\).

The proof is in \Cref{ap:maintheoremproof}. The theorem supplies a compact way of understanding the equilibrium. At rank \(1-e^{-t}\), the agent tolerates a drawdown of size \(\delta_\mu(t)\). Because \(m_\mu'(t)=\delta_\mu(t)\), the same object also determines how spread out equilibrium outcomes are across adjacent ranks. The function \(b(t)\) is the marginal rank incentive in this coordinate. The factor \(y'\left(1-e^{-t}\right)\) is the marginal utility of rank, while the factor \(e^{-t}\) converts a change in the hazard coordinate into a change in rank. Since \(b(t)\to0\), rank incentives vanish at the top of the distribution, and the equilibrium drawdown converges to the single-agent benchmark \(\bar d_\mu\).

\section{Comparative Statics}\label{sec:comparative-statics}

This section contains the main comparative statics. We compare economies that differ only in perceived drift \(\mu\). Crucially we evaluate the underlying process and all realized outcomes under the true law \(\mathbb P^0\). The characterization in \Cref{sec:eq-characterization} makes the comparison transparent. Each equilibrium is summarized by the quantile-indexed drawdown \(\delta_\mu\), and the equilibrium quantile map satisfies \(m_\mu'(t)=\delta_\mu(t)\). Thus, once we know how optimism changes \(\delta_\mu\), we immediately know how it changes the distribution of stopped maxima.

The main force is simple. A more optimistic agent believes continued experimentation is more promising, so she is willing to continue after a larger decline from her personal best. In equilibrium, this greater tolerance for setbacks appears at every finite hazard-time quantile.
\begin{lemma}\label{lem:q-mono}
Let \(\mu_1,\mu_2\in\mathcal M\) with \(\mu_1<\mu_2\). Then \(\delta_{\mu_2}(t)>\delta_{\mu_1}(t)\) for every finite \(t\ge0\). Equivalently, \(m_{\mu_2}'(t)>m_{\mu_1}'(t)\) for every finite \(t\ge0\).
\end{lemma}

The proof is in the appendix. The intuition comes directly from the equation in \Cref{thm:equilibrium-characterization}. The rank term \(b(t)=e^{-t}y'\left(1-e^{-t}\right)\) is the same across economies, while a larger \(\mu\) raises the perceived value of continued search and also raises the limiting single-agent drawdown \(\bar d_\mu\). Our argument in the appendix shows that this ordering cannot be undone at any finite \(t\). Since \(m_\mu'(t)=\delta_\mu(t)\), the more optimistic economy has a steeper true-population quantile map everywhere.

Now let \(M_\mu\coloneqq M_{\tau_{d_\mu}}\) denote the equilibrium stopped maximum under the true law. The previous lemma compares quantile slopes. The next lemma converts that comparison into a distributional statement by representing stopped maxima directly in quantile coordinates.

\begin{lemma}\label{lem:exp-quantile}
Let \(Z\sim\mathrm{Exp}(1)\). Then under \(\mathbb P^0\), \(M_\mu\stackrel{d}{=}m_\mu(Z)\).
\end{lemma}

This follows immediately from \(K_\mu(m_\mu(t))=1-e^{-t}\). The random variable \(Z\) is the hazard-time rank of a randomly chosen agent in the true population, and \(m_\mu(Z)\) converts that rank into the realized stopped maximum. This representation lets us compare different economies using the same random rank draw \(Z\).

\begin{theorem}\label{thm:mean-var-mu}
Let \(\mu_1,\mu_2\in\mathcal M\) with \(\mu_1<\mu_2\). In the unique regular drawdown equilibria,
\[
M_{\mu_2}\succeq_{\mathrm{FOSD}}M_{\mu_1},
\qquad \text{and}
\qquad
\operatorname{Var}^0(M_{\mu_2})>\operatorname{Var}^0(M_{\mu_1}).
\]
\end{theorem}

The proof is in the appendix. The idea is short. By \Cref{lem:q-mono}, \(m_{\mu_2}\) is everywhere steeper than \(m_{\mu_1}\). Since both quantile maps start at zero, \(m_{\mu_2}(t)>m_{\mu_1}(t)\) for every \(t>0\). Using the representation in \Cref{lem:exp-quantile}, the more optimistic economy, therefore, has a higher stopped maximum at every positive quantile, which drives first-order stochastic dominance and a higher mean.

The variance result comes from the same geometry. Let \(\Delta(t)\coloneqq m_{\mu_2}(t)-m_{\mu_1}(t)\). Since \(m_{\mu_2}'(t)>m_{\mu_1}'(t)\), the gain \(\Delta(t)\) is increasing in \(t\). Under the common-\(Z\) representation, the more optimistic economy can be written as \(m_{\mu_1}(Z)+\Delta(Z)\). Thus, the optimism-induced gain is largest precisely for agents who are already high in the distribution. Optimism raises the whole distribution, but it raises the upper quantiles by more in absolute terms, so dispersion increases.

The stopped maximum \(M_\mu\) is the output object compared above, but the same equilibrium boundary also determines how long agents experiment in true time. Define
\[
\tau_\mu\coloneqq\tau_{d_\mu}=\inf\left\{t\ge0\colon M_t-X_t\ge d_\mu(M_t)\right\}.
\]
In the appendix, we derive the useful identity
\[
\mathbb E^0[\tau_\mu]=\frac{1}{\sigma^2}\mathbb E^0\left[d_\mu(M_\mu)^2\right]=\frac{1}{\sigma^2}\int_0^\infty\delta_\mu(t)^2e^{-t}dt.
\]
This formula clarifies the source of the output effect. Optimists do not obtain higher stopped maxima because the true process has a higher drift; all realized comparisons are under \(\mathbb P^0\). They obtain higher stopped maxima because they tolerate larger drawdowns and so search longer. The weight \(e^{-t}\) is the density of the hazard-time rank \(Z\), while the term \(\delta_\mu(t)^2/\sigma^2\) reflects the quadratic time scale of Brownian fluctuations.

\begin{proposition}\label{prop:true-stopping-times}
If \(\mu_1,\mu_2\in\mathcal M\) and \(\mu_1<\mu_2\), then
\(\mathbb E^0\left[\tau_{\mu_2}\right]>\mathbb E^0\left[\tau_{\mu_1}\right]\).
\end{proposition}

The result follows from the identity above and \Cref{lem:q-mono}: higher optimism raises \(\delta_\mu(t)\) at every finite \(t\), so it raises the population average of squared drawdown tolerances. Thus, the same behavioral change that raises stopped maxima also lengthens true experimentation time.

Together, these results capture both the benefit and the hazard of optimism. Greater optimism raises realized output under the true law, but it does so by lengthening search and steepening the equilibrium quantile function. The gains are, therefore, uneven across ranks: the whole distribution improves, while dispersion increases. On to the empirics.


\clearpage

\section{Empirical Evidence}\label{sec:empirics}



\begin{figure}[h!]
\centering
\includegraphics[width=0.48\linewidth]{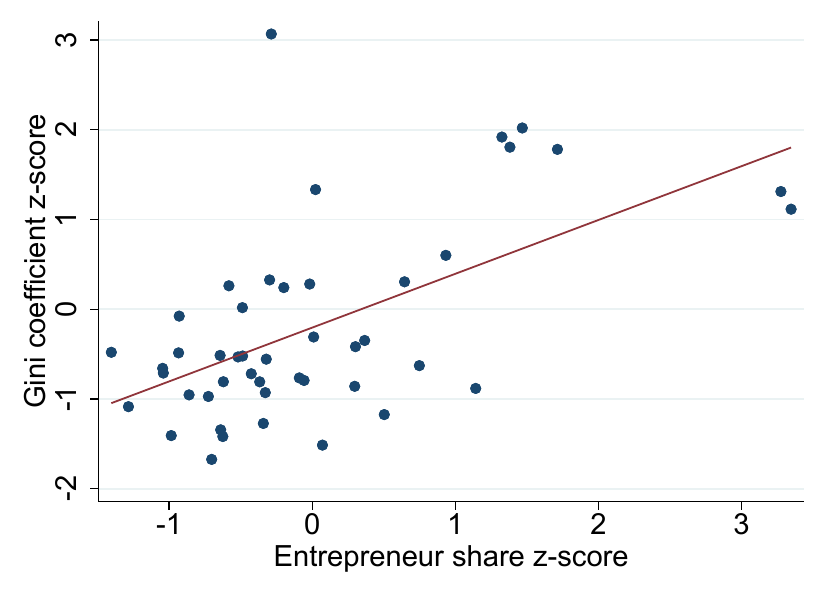}
\includegraphics[width=0.48\linewidth]{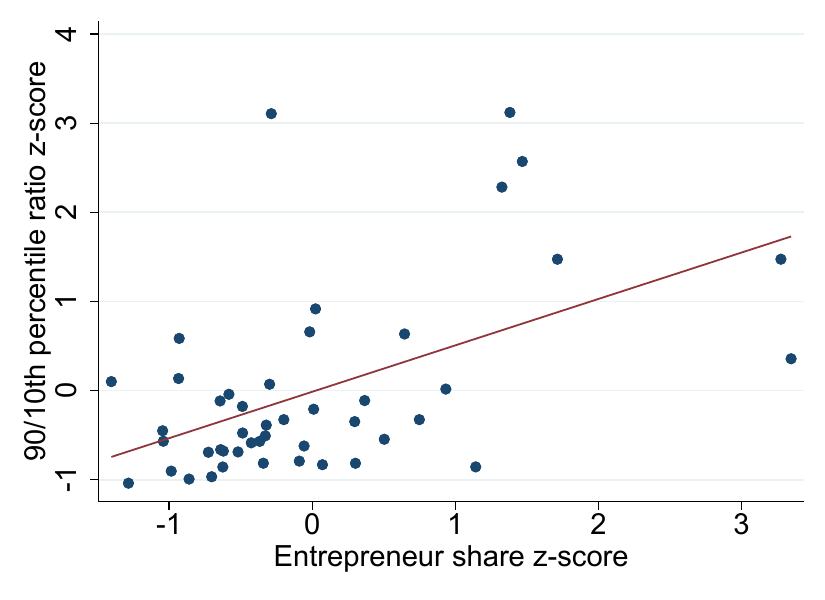}
\caption{Entrepreneurship and Inequality} 
\label{fig:fig1}
\vspace{-8mm}
\begin{singlespace} \begin{flushleft} \footnotesize These figures show country-level correlations between standardized measures of inequality as measured by the Gini index (left panel) or 90/10th percentile ratio (right panel).
\end{flushleft} \end{singlespace}
\end{figure}

\clearpage



\begin{landscape}

\begin{table}[h!]
\centering 
{
\def\sym#1{\ifmmode^{#1}\else\(^{#1}\)\fi}
\begin{tabular}{l*{8}{c}}
\hline\hline
                    &\multicolumn{4}{c}{Gini coefficient}                           &\multicolumn{4}{c}{90/10th percentile ratio}                   \\
                    &\multicolumn{1}{c}{(1)}   &\multicolumn{1}{c}{(2)}   &\multicolumn{1}{c}{(3)}   &\multicolumn{1}{c}{(4)}   &\multicolumn{1}{c}{(5)}   &\multicolumn{1}{c}{(6)}   &\multicolumn{1}{c}{(7)}   &\multicolumn{1}{c}{(8)}   \\
\hline
Entrepreneurship share&       0.599***&       0.454** &       0.620***&       0.469** &       0.521***&       0.379*  &       0.531***&       0.386*  \\
                    &     (0.130)   &     (0.181)   &     (0.123)   &     (0.178)   &     (0.133)   &     (0.200)   &     (0.132)   &     (0.202)   \\
[1em]
Log GDP per capita  &               &      -0.409** &               &      -0.334*  &               &      -0.223   &               &      -0.187   \\
                    &               &     (0.185)   &               &     (0.189)   &               &     (0.205)   &               &     (0.215)   \\
[1em]
Log population      &               &               &       0.186** &       0.123   &               &               &       0.099   &       0.059   \\
                    &               &               &     (0.074)   &     (0.085)   &               &               &     (0.079)   &     (0.096)   \\
\hline
Observations        &          46   &          34   &          46   &          34   &          46   &          34   &          46   &          34   \\
R\textsuperscript{2}&       0.327   &       0.350   &       0.415   &       0.392   &       0.259   &       0.183   &       0.285   &       0.193   \\
\hline\hline
\end{tabular}
}

\caption{Entrepreneurship and Inequality}
\label{tab:new_entrep_regs}
\vspace{-8mm}
\begin{singlespace} \begin{flushleft} \footnotesize This table shows a positive relationship between inequality and entrepreneurship after controlling for GDP per capita and population.
\end{flushleft} \end{singlespace}
\end{table} 

\end{landscape}

\clearpage


\begin{table}[h!]
\centering 
{
\def\sym#1{\ifmmode^{#1}\else\(^{#1}\)\fi}
\begin{tabular}{l*{4}{c}}
\hline\hline
                    &\multicolumn{1}{c}{(1)}   &\multicolumn{1}{c}{(2)}   &\multicolumn{1}{c}{(3)}   &\multicolumn{1}{c}{(4)}   \\
\hline
Willingness to take risks&       0.064   &               &               &               \\
                    &     (0.172)   &               &               &               \\
[1em]
Uncertainty avoidance index&               &       0.168   &               &               \\
                    &               &     (0.169)   &               &               \\
[1em]
Equity share of household assets&               &               &       0.292   &               \\
                    &               &               &     (0.233)   &               \\
[1em]
Lottery sales per capita&               &               &               &      -0.370*  \\
                    &               &               &               &     (0.197)   \\
\hline
Observations        &          33   &          46   &          26   &          34   \\
R\textsuperscript{2}&       0.004   &       0.022   &       0.061   &       0.099   \\
\hline\hline
\end{tabular}
}

\caption{Entrepreneurship and Risk Aversion}
\label{tab:entrep_risk}
\vspace{-8mm}
\begin{singlespace} \begin{flushleft} \footnotesize This table shows that there is no statistically or economically significant relationship between country-level measures of entrepreneurship and risk taking.
\end{flushleft} \end{singlespace}
\end{table} 

\clearpage


\begin{table}[h!]
\centering 
{
\def\sym#1{\ifmmode^{#1}\else\(^{#1}\)\fi}
\begin{tabular}{l*{4}{c}}
\hline\hline
                    &\multicolumn{1}{c}{(1)}   &\multicolumn{1}{c}{(2)}   &\multicolumn{1}{c}{(3)}   &\multicolumn{1}{c}{(4)}   \\
\hline
Entrepreneurship share&       0.593***&       0.598***&       0.505** &       0.623***\\
                    &     (0.165)   &     (0.134)   &     (0.188)   &     (0.159)   \\
[1em]
Willingness to take risks&       0.289   &               &               &               \\
                    &     (0.173)   &               &               &               \\
[1em]
Uncertainty avoidance index&               &       0.078   &               &               \\
                    &               &     (0.162)   &               &               \\
[1em]
Lottery sales per capita&               &               &      -0.309   &               \\
                    &               &               &     (0.220)   &               \\
[1em]
Equity share of household assets&               &               &               &       0.171   \\
                    &               &               &               &     (0.187)   \\
\hline
Observations        &          32   &          44   &          34   &          26   \\
R\textsuperscript{2}&       0.361   &       0.344   &       0.293   &       0.454   \\
\hline\hline
\end{tabular}
}

\caption{Entrepreneurship, Inequality, and Risk Aversion}
\label{tab:ineq_entrep_control_risk}
\vspace{-8mm}
\begin{singlespace} \begin{flushleft} \footnotesize This table shows that the relationship between inequality and entrepreneurship holds even after controlling for measures of risk preferences.
\end{flushleft} \end{singlespace}
\end{table} 

\clearpage


\begin{table}[h!]
\centering 
\small
{
\def\sym#1{\ifmmode^{#1}\else\(^{#1}\)\fi}
\begin{tabular}{l*{5}{c}}
\hline\hline
                    &\multicolumn{1}{c}{(1)}   &\multicolumn{1}{c}{(2)}   &\multicolumn{1}{c}{(3)}   &\multicolumn{1}{c}{(4)}   &\multicolumn{1}{c}{(5)}   \\
\hline
Belief in own abilities&       0.561***&               &               &               &       0.538***\\
                    &     (0.120)   &               &               &               &     (0.198)   \\
[1em]
Others think you are innovative&               &       0.406***&               &               &      -0.033   \\
                    &               &     (0.132)   &               &               &     (0.221)   \\
[1em]
Think entrepreneur is a good career choice&               &               &       0.320** &               &       0.025   \\
                    &               &               &     (0.140)   &               &     (0.153)   \\
[1em]
Know someone who recently started a business&               &               &               &       0.373***&       0.085   \\
                    &               &               &               &     (0.134)   &     (0.157)   \\
\hline
Observations        &          50   &          50   &          49   &          50   &          49   \\
R\textsuperscript{2}&       0.314   &       0.165   &       0.100   &       0.139   &       0.327   \\
\hline\hline
\end{tabular}
}

\caption{Drivers of Entrepreneurship}
\label{tab:entrep_beliefs}
\vspace{-8mm}
\begin{singlespace} \begin{flushleft} \footnotesize This table shows that the relationship between entrepreneurship and self-reported measures of beliefs and perceptions related to entrepreneurship.
\end{flushleft} \end{singlespace}
\end{table} 



\begin{table}[h!]
\centering 
\small
{
\def\sym#1{\ifmmode^{#1}\else\(^{#1}\)\fi}
\begin{tabular}{l*{5}{c}}
\hline\hline
                    &\multicolumn{1}{c}{(1)}   &\multicolumn{1}{c}{(2)}   &\multicolumn{1}{c}{(3)}   &\multicolumn{1}{c}{(4)}   &\multicolumn{1}{c}{(5)}   \\
\hline
Fear of failure     &       0.025   &               &               &               &       0.005   \\
                    &     (0.144)   &               &               &               &     (0.147)   \\
[1em]
Motivation for great wealth&               &      -0.126   &               &               &      -0.252   \\
                    &               &     (0.143)   &               &               &     (0.163)   \\
[1em]
Assign high social status to entrepreneurs&               &               &       0.112   &               &       0.035   \\
                    &               &               &     (0.146)   &               &     (0.153)   \\
[1em]
Business failure rate&               &               &               &       0.185   &       0.280   \\
                    &               &               &               &     (0.142)   &     (0.168)   \\
\hline
Observations        &          50   &          50   &          49   &          50   &          49   \\
R\textsuperscript{2}&       0.001   &       0.016   &       0.012   &       0.034   &       0.088   \\
\hline\hline
\end{tabular}
}

\caption{Preferences and Entrepreneurship}
\label{tab:entrep_prefs}
\vspace{-8mm}
\begin{singlespace} \begin{flushleft} \footnotesize This table shows that entrepreneurship has no meaningful relationship with several alternative measures of risk preferences, including actual business failure rates.
\end{flushleft} \end{singlespace}
\end{table} 

\clearpage


\bibliography{sample.bib}

\appendix

\section{Equilibrium Characterization \texorpdfstring{(\Cref{sec:eq-characterization}) Proofs}{Equilibrium Characterization Proofs}}\label[appendix]{omittedproofs}

\subsection{Proof of \texorpdfstring{\Cref{lem:survival}}{Theorem~\ref{lem:survival}}}\label{app:firstlem}

\begin{proof}[Proof of \Cref{lem:survival}] Define \(G(x)\coloneqq\exp\left(\int_0^x\frac{1}{d(r)}dr\right)\) and \(g(x)\coloneqq G'(x)=G(x)/d(x)\). Recall that \(D_t\coloneqq M_t-X_t\), and consider \[N_t\coloneqq G(M_t)-g(M_t)D_t.\]
\begin{claim}
    \(N\) is a local martingale under \(\mathbb P^0\).
\end{claim}
\begin{proof}
    The running maximum \(M_t\) can increase only when \(X_t=M_t\), i.e., only when \(D_t=0\). Equivalently, the Stieltjes measure \(dM_t\) assigns mass only to times at which \(D_t=0\), and so \(D_t dM_t=0\). Since \(dD_t=dM_t-dX_t\), we have
    \[dN_t=G'(M_t)dM_t-g(M_t)dD_t-D_tg'(M_t)dM_t=g(M_t)dX_t,\] so \(N\) is a local martingale under \(\mathbb P^0\).
\end{proof}
Next, we establish the first conclusion in the lemma.
\begin{claim}\label{cl:lem1no2}
    For every \(m\ge0\), \(\mathbb P^0\left(T_m<\tau_d\right)=\exp\left(-A_d(m)\right)\).
\end{claim}
\begin{proof}
Stop \(N\) at \(T_m\wedge\tau_d\). On this stopped interval, \(M_t\le m\) and \(0\le D_t\le d(M_t)\), so \(0\le N_{t\wedge T_m\wedge\tau_d}\le G(m)\). Accordingly, the stopped local martingale is bounded, hence, a true martingale; and optional stopping yields \(\mathbb E^0\left[N_{t\wedge T_m\wedge\tau_d}\right]=N_0\). Since \(T_m<\infty\) \(\mathbb P^0\)-almost surely, letting \(t\to\infty\), we obtain \(\mathbb E^0\left[N_{T_m\wedge\tau_d}\right]=N_0\). On \(\left\{T_m<\tau_d\right\}\), the continuous process \(X\) reaches the new record level \(m\) exactly at time \(T_m\). Hence, \(M_{T_m}=X_{T_m}=m\), so \(D_{T_m}=0\) and \(N_{T_m}=G(m)\). On \(\left\{\tau_d<T_m\right\}\), \(D_{\tau_d}=d(M_{\tau_d})\), so \(N_{\tau_d}=0\). Since \(d>0\), the two stopping times cannot coincide. Accordingly,
\[
1=N_0=\mathbb E^0\left[N_{T_m\wedge\tau_d}\right]=G(m)\mathbb P^0\left(T_m<\tau_d\right).
\]
Thus, \(\mathbb P^0(T_m<\tau_d)=G(m)^{-1}=\exp\left(-A_d(m)\right)\).
\end{proof}
We finish with the last two conclusions of the result.

\begin{claim}
     If \(A_d(m)\to\infty\), then \(\tau_d<\infty\) \(\mathbb P^0\)-almost surely.
\end{claim}
\begin{proof}Under \(\mathbb P^0\), \(M_t\to\infty\) almost surely. Therefore, \(\left\{\tau_d=\infty\right\}=\bigcap_{n\ge1}\left\{T_n<\tau_d\right\}\), and \Cref{cl:lem1no2}  delivers
\[
\mathbb P^0\left(\tau_d=\infty\right)=\lim_{n\to\infty}\mathbb P^0\left(T_n<\tau_d\right)=\lim_{n\to\infty}\exp\left(-A_d(n)\right)=0. \qedhere
\]
\end{proof}

\begin{claim}
    The induced distribution of \(M_{\tau_d}\) has CDF \(K_d(m)\coloneqq1-\exp\left(-A_d(m)\right)\).
\end{claim}
\begin{proof}
    Since \(\tau_d<\infty\) almost surely, \(\left\{T_m<\tau_d\right\}=\left\{M_{\tau_d}\ge m\right\}\). Indeed, \(T_m<\tau_d\) implies \(M_{\tau_d}\ge m\), while \(M_{\tau_d}\ge m\) implies \(T_m\le\tau_d\), and equality is impossible by the no-coincidence argument above. Therefore, \(\mathbb P^0\left(M_{\tau_d}\ge m\right)=\exp\left(-A_d(m)\right)\). Since \(A_d\) is continuous, this survival function is continuous, so \(M_{\tau_d}\) has no atoms. Accordingly, \(\mathbb P^0\left(M_{\tau_d}\le m\right)=1-\exp\left(-A_d(m)\right)\).
\end{proof}
    This proves the lemma.\end{proof}

\subsection{Two Auxiliary Results}\label{app:aux}

\Cref{lem:survival} maps a drawdown boundary into the true-law distribution of stopped maxima. We now go the other direction. First, in \Cref{prop:fixed-K-verification}, for a fixed conjectured rank schedule \(K\), we characterize a boundary that is optimal for the agent under the perceived law. Then, in \Cref{lem:free-boundary-necessity} we show conversely that any (regular) optimal boundary indeed satisfies this characterization.


Recall that the effective prize is \(\phi_K(m)\coloneqq m+y\left(K(m)\right)\). The verification theorem is stated for a generic prize function \(\phi\). Write \(D_t\coloneqq M_t-X_t\). By Skorokhod reflection \citet[Ch.~3, Sec.~6.C]{KaratzasShreve1991}, \[\label{eq:refldyn}\tag{\(1\)}D_t=D_0-\mu t-\sigma W_t^\mu+L_t \qquad \text{and} \qquad M_t=M_0+L_t,\] where \(L_t\) increases only when \(D_t=0\).
\begin{proposition}\label{prop:fixed-K-verification}
Fix \(\mu\in(0,c)\), and let \(\phi\in C^1\left([0,\infty)\right)\) be increasing and bounded below, with \(\sup_m\phi^\prime(m)<c/\mu\). Suppose an agent's drawdown boundary \(d\) is \(C^1\), positive, bounded, bounded away from zero, and satisfies \(d(m)\to d_\infty<\infty\). If \(d\) solves
\[\label{eq:boundaryODE}\tag{\(2\)}
d^\prime(m)=1-\frac{\mu}{c}\frac{\phi^\prime(m)}{1-\exp\left(-\frac{2\mu}{\sigma^2}d(m)\right)},
\]
then \(\tau_d\) is optimal from every \((m,z)\) with \(0\le z<d(m)\), among stopping times \(\tau\) satisfying \(\mathbb E^\mu_{m,z}\left[\tau\right]<\infty\) and \(\mathbb E^\mu_{m,z}\left[\phi\left(M_\tau\right)^+\right]<\infty\). Moreover, the associated fixed-\(\phi\) value function is regular for \(d\) with terminal payoff \(\phi\) at perceived drift \(\mu\). If, in addition, \(\phi^\prime(m)\to1\), then \(d_\infty=\bar d_\mu\).
\end{proposition}
\begin{proof}[Proof of \Cref{prop:fixed-K-verification}]
Let \(a\coloneqq2\mu/\sigma^2\). For \(0\le z\le d\left(m\right)\), define
\[
V\left(m,z\right)\coloneqq\phi\left(m\right)+\frac{c}{\mu}\left[d\left(m\right)-z-\frac{1-\exp\left(-a\left(d\left(m\right)-z\right)\right)}{a}\right],
\]
and for \(z\ge d\left(m\right)\), set \(V\left(m,z\right)\coloneqq\phi\left(m\right)\). Again, we proceed by going through a sequence of claims. Throughout this proof, admissible means satisfying the two integrability conditions in the statement.
\begin{claim}
    The function \(V\) is \(C^{1,2}\left(C_d\right)\), extends continuously with the relevant one-sided derivatives to the boundaries, and satisfies \(V\ge\phi\), the HJB equation in the continuation region, value matching, smooth fit, and the reflected-boundary condition.
\end{claim}
\begin{proof}
Inside \(C_d\), we differentiate to obtain
\[
V_m\left(m,z\right)=\phi^\prime\left(m\right)+\frac{c}{\mu}d^\prime\left(m\right)\left[1-\exp\left(-a\left(d\left(m\right)-z\right)\right)\right],
\]
\[
V_z\left(m,z\right)=\frac{c}{\mu}\left[\exp\left(-a\left(d\left(m\right)-z\right)\right)-1\right],
\quad \text{and} \quad
V_{zz}\left(m,z\right)=\frac{2c}{\sigma^2}\exp\left(-a\left(d\left(m\right)-z\right)\right).
\]
Since \(\phi,d\in C^1\), these derivatives are continuous on \(C_d\), and the relevant one-sided boundary derivatives exist continuously. Since \(x-\left(1-\exp\left(-ax\right)\right)/a\ge0\) for \(x\ge0\), \(V\ge\phi\). At \(z=d\left(m\right)\), \(V\left(m,d\left(m\right)\right)=\phi\left(m\right)\) and \(V_z\left(m,d\left(m\right)^-\right)=0\). Using the expressions for \(V_z\) and \(V_{zz}\), we also have \(\left(\sigma^2/2\right)V_{zz}-\mu V_z-c=0\) on \(C_d\). At the reflecting boundary \(z=0\),
\[
V_m\left(m,0^+\right)+V_z\left(m,0^+\right)=\phi^\prime\left(m\right)+\frac{c}{\mu}\left[1-\exp\left(-ad\left(m\right)\right)\right]\left[d^\prime\left(m\right)-1\right].
\]
Thus, the boundary ODE \eqref{eq:boundaryODE} is exactly the condition \(V_m\left(m,0^+\right)+V_z\left(m,0^+\right)=0\).
\end{proof}

\begin{claim}
    The process \(V\left(M_t,D_t\right)-ct\) is a local supermartingale under \(\mathbb P^\mu\).
\end{claim}
\begin{proof}
Using the reflected dynamics stated in \eqref{eq:refldyn}, the local-time term at \(z=0\) is proportional to \(V_m+V_z\). The generalized Itô formula with local time on surfaces \cite[Theorem~2.1]{Peskir2007Surfaces} applies to the pasted value function. The HJB kills the drift in the continuation region, smooth fit kills the boundary local time at \(z=d\left(m\right)\), and the condition \(V_m\left(m,0\right)+V_z\left(m,0\right)=0\) kills the reflection term at \(z=0\). In the stopping region \(z>d\left(m\right)\), the pasted value is \(V\left(m,z\right)=\phi\left(m\right)\), so \(V_z=V_{zz}=0\), \(M_t\) is locally constant away from \(z=0\), and the drift of \(V\left(M_t,D_t\right)-ct\) is \(-c\). Hence, \(V\left(M_t,D_t\right)-ct\) is a local supermartingale.
\end{proof}

\begin{claim}\label{cl:prop323}
    For every admissible stopping time \(\tau\),
    \[
    \mathbb E^\mu_{m,z}\left[V\left(M_\tau,D_\tau\right)-c\tau\right]\le V\left(m,z\right).
    \]
\end{claim}
\begin{proof}
Fix an arbitrary starting state \(\left(m,z\right)\) with \(0\le z<d\left(m\right)\). Let \[\eta_N\coloneqq\inf\left\{t\ge0\colon M_t+D_t\ge N\right\} \qquad \text{and} \qquad \tau_N\coloneqq\tau\wedge N\wedge\eta_N.\] For \(N>m+z\), the stopping time \(\tau_N\) is bounded by \(N\), and before \(\eta_N\) we have \(M_t+D_t\le N\). Thus, \(V\left(M_{t\wedge\tau_N},D_{t\wedge\tau_N}\right)-c\left(t\wedge\tau_N\right)\) is bounded. Consequently, from optional sampling
\[
\mathbb E^\mu_{m,z}\left[V\left(M_{\tau_N},D_{\tau_N}\right)-c\tau_N\right]\le V\left(m,z\right).
\]
Because \(d\) is bounded, \(V\le\phi+c\sup_m d\left(m\right)/\mu\), while \(V\ge\phi\) and \(\phi\) is bounded below. Because \(\tau_N\uparrow\tau\) and the sample paths are continuous, \(M_{\tau_N}\to M_\tau\) and \(D_{\tau_N}\to D_\tau\). 

Let \(C\coloneqq c\sup_m d\left(m\right)/\mu\). Since \(V\le\phi+C\), \(M_{\tau_N}\le M_\tau\), and \(\phi\) is increasing, the positive parts are bounded by the integrable random variable \(\phi\left(M_\tau\right)^++C\). Since \(V\ge\phi\) and \(\phi\) is bounded below, \(V\left(M_{\tau_N},D_{\tau_N}\right)-c\tau_N\ge \inf_m\phi\left(m\right)-c\tau\), whose right-hand side is integrable. Fatou's lemma, therefore, yields
\[
\mathbb E^\mu_{m,z}\left[V\left(M_\tau,D_\tau\right)-c\tau\right]
\le
\liminf_{N\to\infty}
\mathbb E^\mu_{m,z}\left[V\left(M_{\tau_N},D_{\tau_N}\right)-c\tau_N\right]
\le V\left(m,z\right).\qedhere
\]
\end{proof}

Since \(V\ge\phi\), \Cref{cl:prop323} implies
\[
\mathbb E^\mu_{m,z}\left[\phi\left(M_\tau\right)-c\tau\right]\le V\left(m,z\right)
\]
for every admissible stopping time \(\tau\). It remains to show that \(\tau_d\) is admissible and attains this upper bound.

\begin{claim}
    The stopping time \(\tau_d\) is admissible.
\end{claim}
\begin{proof}
Let \(D\coloneqq\sup_m d\left(m\right)\), and define \(\tau_D\coloneqq\inf\left\{t\ge0\colon D_t\ge D\right\}\). Since \(d\le D\), we have \(\tau_d\le\tau_D\). By Lehoczky's fixed-drawdown formula \cite[Eq.~(5)]{Lehoczky1977}, when the process starts with no drawdown, \(D_0=0\),
\[
\mathbb P^\mu\left(M_{\tau_D}-M_0\ge r\right)=\exp\left(-\rho_Dr\right),
\qquad \text{where} \quad
\rho_D\coloneqq\frac{2\mu}{\sigma^2\left[\exp\left(2\mu D/\sigma^2\right)-1\right]}.
\]
Starting from any drawdown \(z\in\left[0,D\right]\) can only reduce the chance of reaching new maxima before hitting \(D\), so the same exponential tail bound applies.

Moreover, the killed reflected drawdown process on \(\left[0,D\right]\) has finite mean. Indeed, if \(h\left(z\right)\coloneqq\mathbb E^\mu_z\left[\tau_D\right]\), then
\[
\frac{\sigma^2}{2}h^{\prime\prime}\left(z\right)-\mu h^\prime\left(z\right)=-1,\qquad \text{with} \qquad h^\prime\left(0\right)=0,\quad \text{and} \quad h\left(D\right)=0.
\]
Solving this,
\[
\mathbb E^\mu_z\left[\tau_D\right]=\frac{\sigma^2}{2\mu^2}\left[\exp\left(\frac{2\mu D}{\sigma^2}\right)-\exp\left(\frac{2\mu z}{\sigma^2}\right)\right]-\frac{D-z}{\mu}<\infty.
\]
Thus, \(\mathbb E^\mu_{m,z}\left[\tau_D\right]<\infty\), and since \(\tau_d\le\tau_D\), \(\mathbb E^\mu_{m,z}\left[\tau_d\right]<\infty\). Since \(M_{\tau_d}\le M_{\tau_D}\), and since \(\sup_m\phi^\prime\left(m\right)<\infty\) makes \(\phi\) Lipschitz, \(\mathbb E^\mu_{m,z}\left[\phi\left(M_{\tau_d}\right)^+\right]<\infty\). Hence, \(\tau_d\) is admissible.\end{proof}

\begin{claim}
    The stopping time \(\tau_d\) attains the upper bound: 
    \[
    V\left(m,z\right)=\mathbb E^\mu_{m,z}\left[\phi\left(M_{\tau_d}\right)-c\tau_d\right].
    \]
\end{claim}
\begin{proof}
Before \(\tau_d\), the process remains in the continuation region, and at \(\tau_d\), by value matching, \(V\left(M_{\tau_d},D_{\tau_d}\right)=\phi\left(M_{\tau_d}\right)\). The same localization argument as above, now with equality throughout, yields
\[
V\left(m,z\right)=\mathbb E^\mu_{m,z}\left[V\left(M_{\tau_d},D_{\tau_d}\right)-c\tau_d\right]
=
\mathbb E^\mu_{m,z}\left[\phi\left(M_{\tau_d}\right)-c\tau_d\right]. \qedhere
\]
\end{proof}
The preceding claims show that \(\tau_d\) is optimal and that the candidate \(V\) coincides with the fixed-\(\phi\) value function. Since the candidate is \(C^{1,2}\left(C_d\right)\) and satisfies the HJB equation, value matching, smooth fit, and the reflected-boundary condition, this value function is regular for \(d\) with terminal payoff \(\phi\) at perceived drift \(\mu\).

Finally, suppose \(\phi^\prime\left(m\right)\to1\). Since \(d\left(m\right)\to d_\infty\), the right-hand side of the boundary ODE \eqref{eq:boundaryODE} has a finite limit. Since \(d\) itself has a finite limit, this limiting derivative must be zero. Hence,
\[
0=1-\frac{\mu}{c}\frac{1}{1-\exp\left(-\frac{2\mu}{\sigma^2}d_\infty\right)},
\]
which delivers \(d_\infty=\bar d_\mu\).
\end{proof}
As we noted \Cref{prop:fixed-K-verification} shows that any boundary satisfying \eqref{eq:boundaryODE} is optimal. Our next lemma provides the converse within the regular drawdown class: any regular optimal drawdown boundary must satisfy \eqref{eq:boundaryODE}. This is the step that will allow the scalar ODE to characterize, rather than merely construct, regular drawdown equilibria.

Given a prize function \(\phi\), we write
\[
V^\mu_\phi\left(m,z\right)\coloneqq \sup_\tau \mathbb E^\mu_{m,z}\left[\phi\left(M_\tau\right)-c\tau\right],
\]
where the supremum is over stopping times \(\tau\) satisfying \(\mathbb E^\mu_{m,z}\left[\tau\right]<\infty\) and \(\mathbb E^\mu_{m,z}\left[\phi\left(M_\tau\right)^+\right]<\infty\).
\begin{lemma}\label{lem:free-boundary-necessity}
Fix \(\mu\in\left(0,c\right)\). Let \(\phi\in C^1\left(\left[0,\infty\right)\right)\), and let \(d\in C^1\left(\left[0,\infty\right)\right)\) be positive. Suppose \(\tau_d\) is an optimal drawdown rule for the fixed-\(\phi\) problem and that \(V^\mu_\phi\) is regular for \(d\) with terminal payoff \(\phi\) at perceived drift \(\mu\). Then \(d\) satisfies
\[\label{eq:ode2}\tag{\(3\)}
d'(m)=1-\frac{\mu}{c}\frac{\phi'(m)}{1-\exp\left(-\frac{2\mu}{\sigma^2}d(m)\right)}.
\]
If, in addition, \(\phi'(m)\to1\) and \(d(m)\to d_\infty\in(0,\infty)\), then \(d_\infty=\bar d_\mu\).
\end{lemma}

\begin{proof}[Proof of \Cref{lem:free-boundary-necessity}]

Let \(V\coloneqq V^\mu_\phi\) and maintain the shorthand \(a\coloneqq 2\mu/\sigma^2\).
\begin{claim}
    For each \(m\ge0\) and \(0\le z\le d\left(m\right)\),
    \[V\left(m,z\right)=\phi\left(m\right)+\frac{c}{\mu}\left[d\left(m\right)-z-\frac{1-\exp\left(-a\left(d\left(m\right)-z\right)\right)}{a}\right].
    \]
\end{claim}
\begin{proof}
    Fix \(m\ge0\). Since \(V\) is regular, it satisfies
    \[
    \frac{\sigma^2}{2}V_{zz}\left(m,z\right)-\mu V_z\left(m,z\right)-c=0
    \]
    for \(0<z<d\left(m\right)\). This is a linear ODE in \(z\). Setting \(U\left(z\right)\coloneqq V_z\left(m,z\right)\), we get \(U^\prime\left(z\right)-aU\left(z\right)=2c/\sigma^2\), so \(V_z\left(m,z\right)=\Gamma\left(m\right)e^{az}-c/\mu\) for some function \(\Gamma\). Moreover, smooth fit, \(V_z\left(m,d\left(m\right)^-\right)=0\), implies \(\Gamma\left(m\right)=\left(c/\mu\right)e^{-ad\left(m\right)}\). Hence,
    \[
    V_z\left(m,z\right)=\frac{c}{\mu}\left[\exp\left(-a\left(d\left(m\right)-z\right)\right)-1\right].
    \]
    Integrating from \(z\) to \(d\left(m\right)\) and using value matching, \(V\left(m,d\left(m\right)\right)=\phi\left(m\right)\), produces the stated formula for \(V\left(m,z\right)\).
\end{proof}

\begin{claim}
    The boundary \(d\) satisfies the free-boundary ODE \eqref{eq:ode2}.
\end{claim}
\begin{proof}
    Differentiating the formula for \(V\) in the previous claim with respect to \(m\), and then taking \(z\downarrow0\), delivers
    \[
    V_m\left(m,0^+\right)=\phi^\prime\left(m\right)+\frac{c}{\mu}d^\prime\left(m\right)\left[1-\exp\left(-ad\left(m\right)\right)\right].
    \]
    From the formula for \(V_z\),
    \[
    V_z\left(m,0^+\right)=-\frac{c}{\mu}\left[1-\exp\left(-ad\left(m\right)\right)\right].
    \]
    Imposing the reflected-boundary condition \(V_m\left(m,0^+\right)+V_z\left(m,0^+\right)=0\), we get
    \[
    0=\phi^\prime\left(m\right)+\frac{c}{\mu}\left[1-\exp\left(-ad\left(m\right)\right)\right]\left[d^\prime\left(m\right)-1\right],
    \]
    which we solve for \(d^\prime\left(m\right)\) to obtain the desired ODE \eqref{eq:ode2}.
\end{proof}

\begin{claim}
    If \(\phi^\prime\left(m\right)\to1\) and \(d\left(m\right)\to d_\infty\in\left(0,\infty\right)\), then \(d_\infty=\bar d_\mu\).
\end{claim}
\begin{proof}
    Under these assumptions, the right-hand side of \eqref{eq:ode2} has a finite limit. Since \(d\) has a finite limit, this limiting derivative must be zero; otherwise \(d\) could not converge. Thus,
    \[
    1-\exp\left(-\frac{2\mu}{\sigma^2}d_\infty\right)=\frac{\mu}{c},
    \]
    and, therefore, \(d_\infty=\bar d_\mu\).
\end{proof}

The three claims prove the lemma.
\end{proof}

\subsection{Proof of \texorpdfstring{\Cref{thm:equilibrium-characterization}}{Theorem~\ref{thm:equilibrium-characterization}}}\label{ap:maintheoremproof}

\begin{proof}[Proof of \Cref{thm:equilibrium-characterization}]
The equation in the theorem is the drawdown equation written along population quantiles. To simplify notation, define its right-hand side by
\[
F_\mu(t,x)\coloneqq x-\frac{\mu}{c}\frac{x+b(t)}{1-\exp\left(-\frac{2\mu}{\sigma^2}x\right)},
\]
so that the equation is \(\delta'(t)=F_\mu(t,\delta(t))\). The number \(\bar d_\mu\) is the limiting drawdown level, characterized by
\[
1-\exp\left(-\frac{2\mu}{\sigma^2}\bar d_\mu\right)=\frac{\mu}{c}.
\]
Set \(B\coloneqq y'(0)\), \(\alpha_\mu\coloneqq1+B/\bar d_\mu\), and
\[
d_\mu^+\coloneqq\frac{\sigma^2}{2\mu}\log\left(\frac{c}{c-\mu\alpha_\mu}\right).
\]
Also write \(\theta_\mu\coloneqq\mu/c\), \(a_\mu\coloneqq2\mu/\sigma^2\), and \(g_\mu(x)\coloneqq1/\left(1-\exp\left(-a_\mu x\right)\right)\).

We proceed through six claims. The first four claims solve the scalar terminal-value problem. The last two claims reconstruct the equilibrium and prove uniqueness among regular drawdown equilibria.

\begin{claim}
For every \(T>0\), the finite-terminal problem
\[\label{eq:finite-terminal-delta}\tag{\(4\)}
\delta_T'(t)=F_\mu(t,\delta_T(t)),
\quad \text{with} \quad
\delta_T(T)=\bar d_\mu
\]
has a unique solution on \([0,T]\), and this solution satisfies
\[
\bar d_\mu\le\delta_T(t)\le d_\mu^+
\qquad\text{for all }t\in[0,T].
\]
\end{claim}
\begin{proof}
Since \(\mu\in\mathcal M\), \(\theta_\mu\alpha_\mu<1\), so \(d_\mu^+\) is well defined and satisfies \(1-\exp\left(-a_\mu d_\mu^+\right)=\theta_\mu\alpha_\mu\). At the lower endpoint,
\[
F_\mu(t,\bar d_\mu)=-b(t)\le0.
\]
At the upper endpoint,
\[
F_\mu(t,d_\mu^+)=d_\mu^+-\frac{d_\mu^++b(t)}{\alpha_\mu}
=\frac{d_\mu^+\left(\alpha_\mu-1\right)-b(t)}{\alpha_\mu}\ge0,
\]
because \(d_\mu^+\ge\bar d_\mu\), \(\alpha_\mu-1=B/\bar d_\mu\), and \(b(t)\le B\).

Since the limiting condition is imposed at infinity, we approximate it with finite terminal problems. Fix \(T>0\), and let us look for a function \(\delta_T\) on \([0,T]\) satisfying \eqref{eq:finite-terminal-delta}. Equivalently, define \(h_T(r)\coloneqq\delta_T(T-r)\). Then \(h_T(0)=\bar d_\mu\), and \(h_T\) solves
\[
h_T'(r)=-F_\mu(T-r,h_T(r)).
\]

The interval \([\bar d_\mu,d_\mu^+]\) cannot be left by this reversed-time equation. At the lower endpoint, \(h_T=\bar d_\mu\), we have
\[
h_T'(r)=-F_\mu(T-r,\bar d_\mu)=b(T-r)\ge0,
\]
so the solution cannot cross below \(\bar d_\mu\). At the upper endpoint, \(h_T=d_\mu^+\), we have
\[
h_T'(r)=-F_\mu(T-r,d_\mu^+)\le0,
\]
so the solution cannot cross above \(d_\mu^+\). Thus, any solution that starts in \([\bar d_\mu,d_\mu^+]\) remains in this interval.

The right-hand side of the reversed equation is continuous in \(r\) and locally Lipschitz in \(h_T\) whenever \(h_T>0\). As a result, the initial-value problem for \(h_T\) has a unique local solution. The bounds above keep this solution in the compact interval \([\bar d_\mu,d_\mu^+]\), where the right-hand side is bounded. Hence, the solution cannot blow up before time \(T\), so \(h_T\) is defined on all of \([0,T]\). Equivalently, \(\delta_T(t)\coloneqq h_T(T-t)\) is uniquely defined on \([0,T]\), solves
\[
\delta_T'(t)=F_\mu(t,\delta_T(t)),
\qquad \text{and} \qquad
\delta_T(T)=\bar d_\mu,
\]
and satisfies
\[
\bar d_\mu\le\delta_T(t)\le d_\mu^+
\qquad\text{for all }t\in[0,T].\qedhere
\]
\end{proof}

\begin{claim}
There exists a \(C^1\) solution \(\delta_\mu\) of \(\delta_\mu'=F_\mu(t,\delta_\mu)\) on \([0,\infty)\), and it satisfies
\[
\bar d_\mu\le\delta_\mu(t)\le d_\mu^+
\qquad\text{for all }t\ge0.
\]
\end{claim}
\begin{proof}
The family of solutions to \eqref{eq:finite-terminal-delta}, \(\{\delta_T\}_{T>0}\), is uniformly bounded. It is also equi-Lipschitz on compact time intervals, because \(F_\mu\) is bounded on \([0,\infty)\times[\bar d_\mu,d_\mu^+]\). Choose a sequence \(T_n\to\infty\), and pass to a subsequence, still denoted \(T_n\), such that \(T_n\ge n\). Let \(n_j^0\coloneqq j\).

For each integer \(k\ge1\), discard finitely many terms so that all remaining functions are defined on \([0,k]\), apply the Arzelà-Ascoli theorem on \([0,k]\) to the sequence \(\left(\delta_{T_{n_j^{k-1}}}\right)_j\), and choose a subsequence \(\left(n_j^k\right)_j\) such that \(\delta_{T_{n_j^k}}\) converges uniformly on \([0,k]\). Now set \(N_j\coloneqq n_j^j\). For every finite \(R\), the sequence \(\delta_{T_{N_j}}\) converges uniformly on \([0,R]\): if \(k_R\coloneqq\lceil R\rceil\), then for all \(j\ge k_R\), \(N_j\) belongs to the subsequence chosen for \([0,k_R]\). The limits on overlapping intervals agree, and, therefore, define a continuous function \(\delta_\mu\) on \([0,\infty)\).

For fixed \(0\le s\le t\), and all sufficiently large \(j\), the solutions to \eqref{eq:finite-terminal-delta} satisfy
\[
\delta_{T_{N_j}}(t)=\delta_{T_{N_j}}(s)+\int_s^tF_\mu(u,\delta_{T_{N_j}}(u))du.
\]
Because \(\delta_{T_{N_j}}\to\delta_\mu\) uniformly on \([s,t]\), and because \(F_\mu\) is continuous on the compact set \([s,t]\times[\bar d_\mu,d_\mu^+]\), the integrands converge uniformly. Passing to the limit along the selected subsequence produces
\[
\delta_\mu(t)=\delta_\mu(s)+\int_s^tF_\mu(u,\delta_\mu(u))du.
\]
Thus, \(\delta_\mu\) is \(C^1\) and solves \(\delta_\mu'=F_\mu(t,\delta_\mu)\) on \([0,\infty)\). The bounds also pass to the limit, so
\(\bar d_\mu\le\delta_\mu(t)\le d_\mu^+\) for all \(t \geq 0\).\end{proof}

\begin{claim}
The solution \(\delta_\mu\) satisfies \(\delta_\mu(t)\to\bar d_\mu\).
\end{claim}
\begin{proof}
Since \(y^{\prime\prime}\le0\), \(y^\prime(q)\le y^\prime(0)=B\) for \(q\in[0,1]\), and, hence, \(0\le b(t)\le Be^{-t}\le B\). For \(d\in[\bar d_\mu,d_\mu^+]\),
\[
\partial_dF_\mu(t,d)=1-\theta_\mu g_\mu(d)-\theta_\mu\left(d+b(t)\right)g_\mu'(d).
\]
Since \(g_\mu'(d)<0\), \(\theta_\mu g_\mu(d)\le1\) for \(d\ge\bar d_\mu\), and \(b(t)\in[0,B]\), continuity on the compact set \([\bar d_\mu,d_\mu^+]\times[0,B]\) guarantees \(\lambda_\mu>0\) such that \(\partial_dF_\mu(t,d)\ge\lambda_\mu\) for all \(t\ge0\) and all \(d\in[\bar d_\mu,d_\mu^+]\). Let \(x_T(t)\coloneqq\delta_T(t)-\bar d_\mu\). Since \(F_\mu(t,\bar d_\mu)=-b(t)\), we appeal to the mean-value theorem to deduce \(x_T'(t)=a_T(t)x_T(t)-b(t)\) for some \(a_T(t)\ge\lambda_\mu\), with \(x_T(T)=0\). Solving backward,
\[
0\le x_T(t)\le\int_t^T\exp\left(-\int_t^sa_T(v)dv\right)b(s)ds\le\frac{B}{1+\lambda_\mu}e^{-t}.
\]
Passing to the limit along the selected subsequence yields \(0\le\delta_\mu(t)-\bar d_\mu\le Be^{-t}/(1+\lambda_\mu)\), so \(\delta_\mu(t)\to\bar d_\mu\).
\end{proof}

\begin{claim}
The function \(\delta_\mu\) is the unique bounded positive solution satisfying \(\delta(t)\to\bar d_\mu\).
\end{claim}
\begin{proof}
First, any bounded positive solution satisfying \(\delta(t)\to\bar d_\mu\) must lie in \([\bar d_\mu,d_\mu^+]\). If \(d<\bar d_\mu\), then \(1-\exp\left(-a_\mu d\right)<\theta_\mu\), so \(F_\mu(t,d)=d\left(1-\theta_\mu g_\mu(d)\right)-\theta_\mu b(t)g_\mu(d)<0\). Thus, a solution below \(\bar d_\mu\) moves downward and cannot approach \(\bar d_\mu\) from below. The same derivative calculation shows \(F_\mu\) is strictly increasing on \([\bar d_\mu,\infty)\). Hence, if \(d>d_\mu^+\), then \(F_\mu(t,d)>F_\mu(t,d_\mu^+)\ge0\), so a solution above \(d_\mu^+\) cannot move down toward \(\bar d_\mu\). Therefore, every bounded terminal solution lies in the box.

Now let \(\delta\) and \(\tilde\delta\) be two bounded terminal solutions. Their difference \(w=\delta-\tilde\delta\) satisfies \(w'(t)=a(t)w(t)\), where \(a(t)>0\). Hence, for \(T>t\),
\[
w(t)=w(T)\exp\left(-\int_t^Ta(s)ds\right).
\]
Since both solutions converge to \(\bar d_\mu\), \(w(T)\to0\), and so \(w(t)=0\) for every \(t\).
\end{proof}

\begin{claim}
The scalar solution \(\delta_\mu\) reconstructs a regular drawdown equilibrium \(\left(K_\mu,d_\mu\right)\).
\end{claim}
\begin{proof}
Define \(m_\mu(t)\coloneqq\int_0^t\delta_\mu(s)ds\). Since \(\delta_\mu\ge\bar d_\mu>0\), \(m_\mu(t)\to\infty\). Let \(T_\mu\) be the inverse of \(m_\mu\). Define \(K_\mu(m)\coloneqq1-e^{-T_\mu(m)}\) and \(d_\mu(m)\coloneqq\delta_\mu\left(T_\mu(m)\right)\). Then \(\int_0^mdu/d_\mu(u)=T_\mu(m)\), so \Cref{lem:survival} implies that \(K_\mu\) is exactly the true-law distribution induced by \(d_\mu\).

Along the quantile path \(m_\mu(t)\), \(K_\mu(m_\mu(t))=1-e^{-t}\), \(K_\mu'(m_\mu(t))=e^{-t}/\delta_\mu(t)\), and \(\phi_{K_\mu}'(m_\mu(t))=1+b(t)/\delta_\mu(t)\). The reconstructed boundary \(d_\mu\) is \(C^1\), positive, bounded, bounded away from zero, and satisfies \(d_\mu(m)\to\bar d_\mu\). The reconstructed \(K_\mu\) is \(C^1\), so \(\phi_{K_\mu}\) is \(C^1\), increasing, and bounded below. Because \(\delta_\mu(t)\ge\bar d_\mu\),
\[
\frac{\mu}{c}\phi_{K_\mu}'(m_\mu(t))\le\frac{\mu}{c}\left(1+\frac{y'(0)}{\bar d_\mu}\right)<1.
\]
The scalar ODE is exactly the fixed-\(K_\mu\) boundary ODE rewritten along \(m_\mu(t)\), since \(d_\mu'(m_\mu(t))\delta_\mu(t)=\delta_\mu'(t)\). Hence, \(d_\mu\) solves \eqref{eq:boundaryODE} with \(\phi=\phi_{K_\mu}\). Thus, \Cref{prop:fixed-K-verification} applies: \(d_\mu\) is optimal given \(K_\mu\), and the associated fixed-\(\phi_{K_\mu}\) value function is regular for \(d_\mu\) with terminal payoff \(\phi_{K_\mu}\) at perceived drift \(\mu\). Since \(K_\mu\) is generated by \(d_\mu\), the pair \(\left(K_\mu,d_\mu\right)\) is a regular drawdown equilibrium.

Finally, \(K_\mu(m_\mu(t))=1-e^{-t}\), so \(m_\mu(t)=K_\mu^{-1}\left(1-e^{-t}\right)\), and by construction \(m_\mu'(t)=\delta_\mu(t)\).
\end{proof}

\begin{claim}
Every regular drawdown equilibrium at perceived drift \(\mu\) is equal to \(\left(K_\mu,d_\mu\right)\).
\end{claim}
\begin{proof}
Conversely, let \((K,d)\) be any regular drawdown equilibrium. Set \(\phi\coloneqq\phi_K\). Since \(d\) is a regular optimal drawdown boundary for the fixed-\(K\) problem at perceived drift \(\mu\), \(\tau_d\) is optimal for the fixed-\(\phi\) problem and the value function \(V^\mu_K=V^\mu_\phi\) is regular for \(d\) with terminal payoff \(\phi\) at perceived drift \(\mu\). By \Cref{lem:free-boundary-necessity}, \(d\) satisfies the fixed-\(K\) free-boundary ODE. Since \(d\) is bounded above and below, \(A_d(m)\to\infty\), and \Cref{lem:survival} delivers \(K(m)=1-\exp\left(-A_d(m)\right)\). Consequently, \(K\) is \(C^1\), with \(K^\prime(m)=\exp\left(-A_d(m)\right)/d(m)\to0\), so \(\phi_K^\prime(m)=1+y^\prime(K(m))K^\prime(m)\to1\). From \Cref{lem:free-boundary-necessity}, we then have \(\lim_{m\to\infty}d(m)=\bar d_\mu\).

Define \(m(t)\) as the unique point satisfying \(K(m(t))=1-e^{-t}\), and set \(\delta(t)=d(m(t))\). Since \(K(m)=1-\exp\left(-A_d(m)\right)\), we have \(A_d(m(t))=t\), and so \(m^\prime(t)=\delta(t)\). Moreover, \(K^\prime(m(t))=e^{-t}/\delta(t)\), so \(\phi_K^\prime(m(t))=1+b(t)/\delta(t)\). Rewriting the fixed-\(K\) ODE along \(m(t)\) produces the scalar ODE with terminal limit \(\bar d_\mu\). By the uniqueness of the bounded terminal solution, \(\delta=\delta_\mu\). Therefore, \(m=m_\mu\), \(d=d_\mu\), and \(K=K_\mu\).
\end{proof}

The six claims prove the theorem.
\end{proof}

\section{\texorpdfstring{\Cref{sec:comparative-statics} Proofs}{Section~\ref{sec:comparative-statics} Proofs}}\label[appendix]{omittedproofsdeux}

\subsection{Proof of \texorpdfstring{\Cref{lem:q-mono,lem:exp-quantile}}{Lemma~\ref{lem:q-mono}}}

\begin{proof}[Proof of \Cref{lem:q-mono}]
First, \(\bar d_\mu\) is strictly increasing in \(\mu\), since
\[
\bar d_\mu=\frac{\sigma^2}{2c}\frac{\log\left(\frac{1}{1-\mu/c}\right)}{\mu/c},
\]
and \(x\mapsto\log\left(1/(1-x)\right)/x\) is strictly increasing on \((0,1)\).

Let \(w(t)\coloneqq\delta_{\mu_2}(t)-\delta_{\mu_1}(t)\). Since \(\bar d_{\mu_2}>\bar d_{\mu_1}\), \(w(t)>0\) for all sufficiently large \(t\). Suppose, for the sake of contradiction, that \(w(t)\le0\) for some finite \(t\). Since \(w(t)\to\bar d_{\mu_2}-\bar d_{\mu_1}>0\), the set \(\left\{s\ge0\colon w(s)\le0\right\}\) is nonempty and bounded. Let
\[
t^\ast\coloneqq\sup\left\{s\ge0\colon w(s)\le0\right\}.
\]
By continuity, \(w(t^\ast)=0\), and \(w(t)>0\) for all \(t>t^\ast\) sufficiently close to \(t^\ast\). 

Let \(d\coloneqq\delta_{\mu_1}(t^\ast)=\delta_{\mu_2}(t^\ast)\). Subtracting the two scalar ODEs at \(t^\ast\), we obtain
\[
w'(t^\ast)=-\frac{d+b(t^\ast)}{c}\left[\frac{\mu_2}{1-\exp\left(-\frac{2\mu_2}{\sigma^2}d\right)}-\frac{\mu_1}{1-\exp\left(-\frac{2\mu_1}{\sigma^2}d\right)}\right].
\]
The bracketed term is strictly positive because \(x\mapsto x/(1-e^{-x})\) is strictly increasing on \((0,\infty)\), and so \(w'(t^\ast)<0\). But at a last crossing before \(w\) becomes positive to the right, the right derivative cannot be negative. This contradiction proves \(w(t)>0\) for every finite \(t\).
\end{proof}

\begin{proof}[Proof of \Cref{lem:exp-quantile}]
For every \(t\ge0\),
\[
\mathbb P\left(m_\mu(Z)\le m_\mu(t)\right)=\mathbb P(Z\le t)=1-e^{-t}=K_\mu(m_\mu(t)).
\]
Since \(m_\mu\) is continuous and strictly increasing, the claim follows.
\end{proof}

\subsection{Proof of \texorpdfstring{\Cref{thm:mean-var-mu}}{Theorem~\ref{thm:mean-var-mu}}}
\begin{proof}[Proof of \Cref{thm:mean-var-mu}]
By \Cref{lem:q-mono}, \(m_{\mu_2}'(t)>m_{\mu_1}'(t)\) for every finite \(t\). Since \(m_{\mu_1}(0)=m_{\mu_2}(0)=0\), \(m_{\mu_2}(t)>m_{\mu_1}(t)\) for every \(t>0\). \Cref{lem:exp-quantile}, therefore, implies \(M_{\mu_2}\succeq_{\mathrm{FOSD}}M_{\mu_1}\), with strict dominance at every positive quantile.

Let \(\Delta(t)\coloneqq m_{\mu_2}(t)-m_{\mu_1}(t)\). Then \(\Delta\) is increasing and nonconstant. Couple the two stopped maxima using the same \(Z\sim\mathrm{Exp}(1)\), and set \(X\coloneqq m_{\mu_1}(Z)\) and \(Y\coloneqq\Delta(Z)\). Then \(X\stackrel{d}{=}M_{\mu_1}\) and \(X+Y\stackrel{d}{=}M_{\mu_2}\). Since \(X\) and \(Y\) are increasing functions of the same scalar random variable, \(\operatorname{Cov}(X,Y)\ge0\). Indeed, if \(Z'\) is an independent copy of \(Z\), then
\[
2\operatorname{Cov}(X,Y)=\mathbb E\left[\left(m_{\mu_1}(Z)-m_{\mu_1}(Z')\right)\left(\Delta(Z)-\Delta(Z')\right)\right]\ge0.
\]
Since \(Y\) is nonconstant, \(\operatorname{Var}(Y)>0\). Also, \(\delta_\mu(t)\le d_\mu^+\), so \(m_\mu(t)\le d_\mu^+t\), and all second moments are finite. Consequently,
\[
\operatorname{Var}^0(M_{\mu_2})-\operatorname{Var}^0(M_{\mu_1})=\operatorname{Var}(Y)+2\operatorname{Cov}(X,Y)>0.\qedhere
\]
\end{proof}

\subsection{Proof of \texorpdfstring{\Cref{prop:true-stopping-times}}{Theorem~\ref{prop:true-stopping-times}}}

\begin{proof}[Proof of \Cref{prop:true-stopping-times}]
Since \(d_\mu\) is bounded above and below, \(A_\mu\left(m\right)\to\infty\). Hence, \(\tau_\mu<\infty\) \(\mathbb P^0\)-almost surely by \Cref{lem:survival}.

We now compute the joint Laplace transform of \(\tau_\mu\) and \(M_{\tau_\mu}\). In the notation of Avram, Li, and Li (ALL)~\cite[Theorem~2.1 and Corollary~4.2]{AvramLiLi2021}, the drawdown time is defined by a lower boundary \(f\), with \(\tau_f=\inf\left\{t\ge0\colon X_t<f\left(M_t\right)\right\}\). For the present boundary, set \(f_\mu\left(m\right)\coloneqq m-d_\mu\left(m\right)\), so \(m-f_\mu\left(m\right)=d_\mu\left(m\right)\). Since \(d_\mu\) satisfies the free-boundary ODE and \(\phi_{K_\mu}^\prime>0\), we have \(d_\mu^\prime\left(m\right)<1\), so \(f_\mu\) is increasing. The weak inequality in the definition of \(\tau_\mu\) and the strict inequality in the definition of \(\tau_f\) give the same stopping time almost surely. Indeed, at first contact, \(d_\mu\left(M_t\right)>0\), so \(X_t<M_t\). Accordingly, \(M_t\) is locally constant immediately after contact, and the boundary \(f_\mu\left(M_t\right)\) is locally constant as well. By the strong Markov property and the nonstickiness of Brownian motion at a point, the process crosses the boundary immediately.

Specializing the ALL formula to driftless Brownian motion with volatility \(\sigma\), using \(\lambda\left(q\right)\coloneqq\sqrt{2q}/\sigma\), \(\varphi_q^+\left(x\right)=e^{\lambda\left(q\right)x}\) and \(\varphi_q^-\left(x\right)=e^{-\lambda\left(q\right)x}\), produces
\[
\mathbb E^0\left[e^{-q\tau_\mu};M_{\tau_\mu}\in dm\right]
=
\frac{\lambda\left(q\right)}{\sinh\left(\lambda\left(q\right)d_\mu\left(m\right)\right)}
\exp\left[-\int_0^m\lambda\left(q\right)\coth\left(\lambda\left(q\right)d_\mu\left(u\right)\right)du\right]dm.
\]
As a check, when \(d_\mu\left(m\right)\equiv D\), integrating this density against \(e^{-\beta m}\) gives us
\[
\mathbb E^0\left[e^{-q\tau_D-\beta M_{\tau_D}}\right]
=
\frac{\lambda\left(q\right)}
{\lambda\left(q\right)\cosh\left(\lambda\left(q\right)D\right)+\beta\sinh\left(\lambda\left(q\right)D\right)},
\]
which is the classical fixed-drawdown Brownian formula of Taylor and Williams, and the fixed-drawdown specialization of Lehoczky's maximum-based formula \citep{Taylor1975,Williams1976,Lehoczky1977}.

Letting \(q\downarrow0\) in the joint Laplace density yields
\[
\mathbb P^0\left(M_{\tau_\mu}\ge m\right)=e^{-A_\mu\left(m\right)},
\qquad \text{and} \qquad
\mathbb P^0\left(M_{\tau_\mu}\in dm\right)=\frac{e^{-A_\mu\left(m\right)}}{d_\mu\left(m\right)}dm.
\]

The mean stopping time follows by differentiating the Laplace transform at zero. With \(\lambda=\sqrt{2q}/\sigma\),
\[
\lambda\coth\left(\lambda d\right)=\frac1d+\frac{\lambda^2d}{3}+O\left(\lambda^4d^3\right),
\qquad \text{and} \qquad
\frac{\lambda}{\sinh\left(\lambda d\right)}
=
\frac1d\left(1-\frac{\lambda^2d^2}{6}+O\left(\lambda^4d^4\right)\right).
\]
Since \(d_\mu\) is bounded above and below, these expansions are uniform for \(d\) in the range of \(d_\mu\). For \(q\) near zero, the resulting densities and their first \(q\)-derivatives are dominated by \(C\left(1+m\right)\exp\left(-m/C\right)\) for some finite \(C\). Therefore, we can differentiate the expansion at \(q=0\) and integrate over \(m\). Substituting, we have
\[
\mathbb E^0\left[\tau_\mu;M_{\tau_\mu}\in dm\right]
=
\frac{e^{-A_\mu\left(m\right)}}{3\sigma^2d_\mu\left(m\right)}
\left[d_\mu\left(m\right)^2+2\int_0^md_\mu\left(u\right)du\right]dm.
\]
Then divide by the density of \(M_{\tau_\mu}\) to get
\[
\mathbb E^0\left[\tau_\mu\mid M_{\tau_\mu}=m\right]
=
\frac{1}{3\sigma^2}
\left[d_\mu\left(m\right)^2+2\int_0^md_\mu\left(u\right)du\right],
\]
where the conditional expectation is understood as the Radon-Nikodym derivative of \(\mathbb E^0\left[\tau_\mu;M_{\tau_\mu}\in dm\right]\) with respect to the law of \(M_{\tau_\mu}\).

Now let \(B_\mu\left(m\right)\coloneqq\int_0^md_\mu\left(u\right)du\) and \(S_\mu\left(m\right)\coloneqq e^{-A_\mu\left(m\right)}\). Since \(S_\mu^\prime\left(m\right)=-S_\mu\left(m\right)/d_\mu\left(m\right)\), we integrate by parts to obtain
\[
\int_0^\infty B_\mu\left(m\right)\frac{S_\mu\left(m\right)}{d_\mu\left(m\right)}dm
=
\int_0^\infty d_\mu\left(m\right)S_\mu\left(m\right)dm,
\]
where the boundary term vanishes because \(B_\mu\left(m\right)\) grows at most linearly and \(S_\mu\left(m\right)\) decays exponentially. In sum,
\[
\mathbb E^0\left[\tau_\mu\right]
=
\frac{1}{\sigma^2}\int_0^\infty d_\mu\left(m\right)e^{-A_\mu\left(m\right)}dm
=
\frac{1}{\sigma^2}\mathbb E^0\left[d_\mu\left(M_{\tau_\mu}\right)^2\right]
=
\frac{1}{\sigma^2}\int_0^\infty\delta_\mu\left(t\right)^2e^{-t}dt,
\]
where the second equality uses the density of \(M_{\tau_\mu}\) and the third equality follows from the fact that in hazard time, \(A_\mu\left(m_\mu\left(t\right)\right)=t\), \(dm_\mu\left(t\right)=\delta_\mu\left(t\right)dt\), and \(d_\mu\left(m_\mu\left(t\right)\right)=\delta_\mu\left(t\right)\).

We conclude that if \(\mu_1<\mu_2\), then \(\delta_{\mu_2}\left(t\right)>\delta_{\mu_1}\left(t\right)\) for every finite \(t\) by \Cref{lem:q-mono} and so \(\mathbb E^0\left[\tau_{\mu_2}\right]>\mathbb E^0\left[\tau_{\mu_1}\right]\).\end{proof}

\end{document}